\documentclass[conference]{IEEEtran}
\IEEEoverridecommandlockouts
\usepackage{cite}
\usepackage{amsmath,amssymb,amsfonts}
\usepackage{algorithmic}
\usepackage{graphicx}
\usepackage{textcomp}
\usepackage{xcolor}
\usepackage{units}
\usepackage{multirow}
\usepackage{colortbl}
\usepackage{soul}

\usepackage{amsthm}
\usepackage{thmtools, thm-restate}
\usepackage{algorithm}

\newcommand{\hh}{\widehat{h}}

\newcommand{\sigt}{\widetilde{\sigma}}

\newcommand{\w}{\mathbf{w}}

\newcommand{\x}{\mathbf{x}}

\newcommand{\y}{\mathbf{y}}
\newcommand{\z}{\mathbf{z}}

\newcommand{\0}{\mathbf{0}}

\newcommand{\mc}[1]{\mathcal{#1}}
\newcommand{\mb}[1]{\mathbb{#1}}
\newcommand{\mf}[1]{\mathbf{#1}}

\newtheorem{assum}{Assumption}

\newcommand{\alg}{$\mathsf{CHARLES}~$}
\newcommand{\algns}{$\mathsf{CHARLES}$}

\def\BibTeX{{\rm B\kern-.05em{\sc i\kern-.025em b}\kern-.08em
    T\kern-.1667em\lower.7ex\hbox{E}\kern-.125emX}}
\begin{document}

\title{CHARLES: Channel-Quality-Adaptive Over-the-Air Federated Learning over Wireless Networks 
\thanks{This work is supported in part by NSF-2112471 (AI-EDGE). Symbol $^*$ denotes co-primary authors, who have contributed equally to this work.}
}


\newcommand{\kliu}[1]{{\color{red}{\bf Kevin:} #1}}
\newcommand{\jmao}[1]{{\color{blue}{\bf Jiayu:} #1}}


\author{Jiayu Mao$^*$, Haibo Yang$^*$, Peiwen Qiu, Jia Liu, and Aylin Yener
\\Dept. of Electrical and Computer Engineering, The Ohio State University
\\ \{mao.518, yang.9292, qiu.617\}@osu.edu, liu@ece.osu.edu, yener@ece.osu.edu 
}

\maketitle


\begin{abstract}
Over-the-air federated learning (OTA-FL) has emerged as an efficient mechanism that exploits the superposition property of the wireless medium and performs model aggregation for federated learning in the air.
OTA-FL is naturally sensitive to wireless channel fading, which could significantly diminish its learning accuracy.
To address this challenge, in this paper, we propose an OTA-FL algorithm called \alg (\ul{ch}annel-quality-\ul{a}ware ove\ul{r}-the-air \ul{l}ocal \ul{e}stimating and \ul{s}caling).
Our \alg algorithm performs channel state information (CSI) estimation and adaptive scaling to mitigate the impacts of wireless channel fading.
We establish the theoretical convergence rate performance of \alg and analyze the impacts of CSI error on the convergence of \algns. 
We show that the adaptive channel inversion scaling scheme in \alg is robust under imperfect CSI scenarios.
We also demonstrate through numerical results that \alg outperforms existing OTA-FL algorithms with heterogeneous data under imperfect CSI.
\end{abstract}



\vspace{-0.05in}
\section{Introduction} \label{sec:intro}
Fueled by concerns on privacy and communication-efficiency, federated learning (FL) \cite{mcmahan2017communication} has attracted significant attention in recent years and found many applications in practice \cite{mcmahan2021fl}.
FL employs a large number of clients to collaboratively train a global model through relatively infrequent exchanges of model parameters between the server and clients and without sharing any local data of each client.
As a result, FL inherently provides better privacy and communication efficiency as compared to approaches to that communicate user data to the server for training, while also being able to leverage computational parallelism for numerous clients.
However, when deployed in wireless networks, clients performing FL have to deal with the well-known challenges of mobile communications, including limited temporal and spectral resources and power.
To avoid interference between clients, the conventional strategy would be to assign orthogonal either spectral or temporal channel resource blocks to the clients. 
However, in a large-scale FL system, orthogonal transmissions could dramatically decrease spectral efficiency, which in turn could prolong the training. Such a scheme may render FL infeasible in some bandwidth and power-limited scenarios.

To overcome this challenge, over-the-air FL (OTA-FL) has recently been proposed to perform model aggregation for free by utilizing the superposition property of the wireless medium \cite{amiri2020machine}.
Specifically, rather than relying on using orthogonal channels to avoid interference and recover each client's information individually, OTA-FL {\it embraces} interference by allowing all participating clients to simultaneously transmit in the same channel and aggregate over the air directly.
Hence, OTA-FL experiences no resource degradation in communication parallelism, regardless of the number of clients in the system.

Despite these advantages of OTA-FL, a potential disadvantage is that performing model aggregation over the air could be highly susceptible to wireless channel fading effects. In particular, fading could significantly degrade each client's signal (component) in the over-the-air aggregation, particularly in time-varying channels.
To date, most existing works on OTA-FL with channel fading  focus on developing efficient algorithms based on perfect channel state information (CSI).
For example, A-DSGD is proposed in \cite{amiri2020machine} based on an analog scheme with sparsification under simple Gaussian multiple access channels (MAC). 
Fading is also considered in \cite{sery2021over,yang22over}, where perfect CSI is assumed to be available to mitigate channel noise and ensure convergence (see Section~\ref{sec:related} for more in-depth  discussions).
However, perfect CSI is usually difficult to
obtain in practice, especially for fast fading channels\cite{weber2006imperfect}. 
In addition to the impacts of channel fading, FL deployments over wireless networks also face challenges in data and system heterogeneity due to the inherent geo-location diversity of wireless networks (non-i.i.d. and unbalanced dataset sizes; vastly different channel qualities and computation capabilities of clients, etc.). 

To address the aforementioned challenges, in this paper, we propose a new OTA-FL algorithm called \alg (\ul{ch}annel-quality-\ul{a}ware ove\ul{r}-the-air \ul{l}ocal \ul{e}stimating and \ul{s}caling), which performs CSI estimation and adaptive scaling to mitigate the impacts of both channel fading and data/system heterogeneity.
Our main contributions are summarized as follows:
\begin{list}{\labelitemi}{\leftmargin=1em \itemindent=-0.5em \itemsep=.2em}
\item Our proposed {\alg} algorithm allows each client, in a distributed manner, to adaptively determine its {\em transmission power level and number of local update steps} based on its estimated (imperfect) CSI to fully utilize the computation and communication resources.

\item 
We analyze the convergence of {\alg} for  non-convex FL settings and the impacts of CSI estimation error on the convergence of \algns.
We show that the impact of imperfect CSI on stationarity gap convergence can be bounded in terms of channel estimation error.

\item We conduct experiments using convolutional neural network (CNN) on non-i.i.d. MNIST datasets to evaluate the performance of our \alg algorithm.
We show that {\alg} outperforms existing OTA-FL
algorithms.
\end{list}

\begin{table*}[ht]
    \caption{Related Work (`` $\times$ '' means the opposite situation, blank means none).}
    \label{tab:related_work}
    \centering
    \begin{tabular}{|c|c|c|c|c|c|c|c|c|}
         \hline 
         \textbf{Method} & Perfect CSIT & Perfect CSIR& Compression & Non-IID & Analog & Device Schedule \\
         \hline  
         COTAF \cite{sery2021over}&$\surd$& & & $\surd$ & $\surd$ & \\
         \hline
         WMFS \cite{shao2021federated}& &$\times$ & & & $\surd$ & \\
         \hline
         OBDA \cite{zhu2020one1}&$\times$  & & $\surd$ & & $\times$  & $\surd$\\
         \hline
         CA-DSGD \cite{amiri2020federated}& $\surd$ & $\surd$ & $\surd$ & $\surd$ & $\surd$ & \\
         \hline
         Reference\cite{amiri2021blind}& & $\times$ &  & $\surd$ & $\surd$ & \\
         \hline
         Reference\cite{sun2021dynamic}&$\times$ & &  & $\surd$ & $\surd$ & $\surd$\\
         \hline
         QAW-GPR\cite{wadu2021joint}& & $\times$ &  & $\surd$ & $\surd$ & $\surd$\\
         \hline
         FRFL\cite{pase2021convergence}& & $\times$(only CDF) &  & $\surd$ & $\surd$ & $\surd$\\
         \hline
         {\bf \cellcolor{blue!8} CHARLES (This Paper)} & \cellcolor{blue!8} $\times$ & \cellcolor{blue!8} & \cellcolor{blue!8} & \cellcolor{blue!8} $ \cellcolor{blue!8} \surd$ \cellcolor{blue!8} & \cellcolor{blue!8}$\surd$ & \cellcolor{blue!8} \\
         \hline    \end{tabular}
\vspace{-.1in}
\end{table*}

\section{Related Work} \label{sec:related}

OTA-FL in fading channels has been considered in a number of references~\cite{yang2019federated,yang2020federated,zhu2019broadband,sery2020analog,amiri2020federated,amiri2020machine,amiri2019over}. 
Most of these existing works assume perfect CSI at the transmitter side, which allows precoding designs to mitigate channel fading. 
Research on OTA-FL with imperfect CSI has also been ongoing \cite{amiri2020federated,shao2021federated,shao2021bayesian,amiri2021blind,pase2021convergence,zhu2020one1,chen2022over1,sun2021dynamic,wadu2021joint}).
In \cite{amiri2020federated}, the performance of the proposed compressed analog DSGD algorithm under imperfect CSI was considered in experiments only, which shows that imperfect CSI could lead to signal misalignment at the server-side.
In \cite{shao2021federated}, a whitened matched filtering and a sampling scheme is proposed to deal with imperfect CSI on the receiver side.
A Bayesian approach was further proposed in the follow-up work of the same authors\cite{shao2021bayesian}. 
In \cite{amiri2021blind}, a receive beamforming pattern is proposed to compensate the server-side imperfect CSI by equipping the server with a sufficiently large number of antennas. 
Without assuming perfect receiver-side CSI (CSIR), a fixed rate federated learning approach is proposed in \cite{pase2021convergence} using only receiver channel CDF but with slower FL convergence time. 
Instead of assuming perfect transmitter-side CSI (CSIT), one-bit quantization and majority voting was used in \cite{zhu2020one1} to reduce communication cost in FL for digital communication, each user transmits only one bit for each element of the local update vector. 
In \cite{chen2022over1}, MSE of the aggregated global update $\mb{E}||\hat{x}_t - x_t||^2$ ($\hat{x}_t$ is the estimated aggregation that contains CSI) is minimized in each communication round and an alternating optimization approach is used to find optimal receiver beamforming and power control. 
Another line of work on OTA with imperfect CSI related to this paper is based on device scheduling design. 
For example, a dynamic energy-aware scheduling algorithm was proposed in \cite{sun2021dynamic} by taking computation energy constraint into account, while resource allocation with client scheduling was also considered in  \cite{wadu2021joint}. 

We note that the aforementioned existing works only attempted to mitigate the impacts of imperfect CSI from either CSIT or CSIR perspectives.
So far, there remains a lack of theoretical and quantitative understanding on the combined overall impacts of imperfect CSI.
Moreover, the power control component is also decoupled from the FL optimization problem and treated separately in these existing works. 
These factors collectively motivate us to pursue a unified OTA-FL algorithmic design that achieves convergence guarantees under both imperfect CSIT and CSIR.
In our previous work \cite{yang22over}, we have explored a joint computation and power control design under for OTA-FL with perfect CSI. 
Built upon \cite{yang22over}, this work also considers joint computation and power control, but specifically takes into account and counters imperfect CSI by adaptive local update steps.
To highlight our contributions, we summarize the state-of-the-art of OTA-FL with imperfect CSI in Table ~\ref{tab:related_work}.


\section{System Model} \label{sec: prelim}

In this section, we introduce the system model of OTA-FL with imperfect CSI. 

\subsection{Federated Learning Model} \label{subsec: fl}

Consider an FL system with $m$ clients collaboratively training a model coordinated by a server. 
Each client $i$ maintains a local dataset $D_i$ following a distribution $\mc{X}_i$. 
The datasets are assumed to be non-i.i.d. across clients, i.e., $\mc{X}_i \neq \mc{X}_j$ if $ i \neq j, \forall i, j \in [m]$. 
The goal of FL is to minimize a global loss function by finding an optimal model parameter $\x$:

\begin{equation}
    \min_{\x \in \mathbb{R}^d}F(\x) \triangleq \min_{\x\in\mathbb{R}^d} \sum_{i \in [m]} \alpha_i F_i(\x, D_i), 
    \label{eq: objective}
\end{equation}
where $\alpha_i = \frac{| D_i |}{\sum_{i \in [m]} | D_i |}$ is the proportion of the size of local dataset $i$ in the global dataset, $F_i(\x, D_i) \triangleq \frac{1}{| D_i |} \sum_{\xi^i_j \in D_i} F(\x, \xi^i_j)$ denotes the local loss function. In this paper, we assume that $F_i(\x, D_i)$ is non-convex, which is typical in FL. 
In practice, the proportions of different client datasets are typically different, i.e., $\alpha_i \ne \alpha_j$ if $i\ne j$. 

In the $t$-th communication round, the server broadcasts the current global model parameter $\x_t$ to each client. Then, client $i$ starts training local model from $\x_t$ based on its local dataset $D_i$. 
Each client employs stochastic gradient descent (SGD) with initialization $\x^i_{t, 0} = \x_t$ for $\tau_t^i$ steps:
\begin{equation}
    \x^i_{t, k+1} = \x^i_{t, k} - \eta \nabla F_{i}(\x^i_{t, k}, \xi^i_{t, k}), \quad k = 0,\ldots,\tau_t^i-1, \label{sgd}
\end{equation}
where $\xi^i_{t, k}$ is a training sample randomly drawn from $D_i$ in the $k$-th local step in the $t$-th round. 
As in our previous work~\cite{yang22over}, the number of local steps $\tau_t^i$ varies across clients and times. 

The clients upload the model updates to the server once the local training is done. 
Upon receiving all local updates, the server aggregates them and update the global model to $\x_{t+1}$ accordingly. 
Then, the server broadcasts it to the clients and the next communication round $t+1$ starts. 
The training process stops if the global model converges or reaches some predefined limit of iteration number. 
Note that, under OTA-FL,  communication and aggregation happen simultaneously at the server due to the inherent superposition property of the wireless medium.

\subsection{Communication Model} \label{subsec: channel}

For simplicity, we consider a synchronous error-free downlink\footnote{Our results extend to the noisy downlink, which would add an extra error term in convergence bound.} and a wireless fading uplink MAC communication model for the OTA-FL system. 
The server and devices are all equipped with a single antenna. 
We assume that each client receives the global model perfectly, i.e., $\x^i_{t, 0} = \x_t, \forall i\in [m]$. However, for the uplink, devices transmit their update through a shared wireless medium. 
We denote $\z_t^i \in \mb{R}^d$ as the transmitted signal from client $i$, which experiences fading  during transmission. We assume that in each communication round $t$, the uplink channels follow a block fading model, where each channel gain remains constant during transmission of $d$ symbols and changes in the next communication round. 
In this paper, we use analog transmission in OTA-FL to fully utilize the superposition property of MAC. The received signal at the server can be written as:
$\y_t = \sum_{i \in [m]} h_t^i \z^i_t + \w_t$,
where $h_t^i \in \mb{C}$ denotes the channel gain from client $i$ to the server in round $t$, $\w_t$ represents the i.i.d. additive white Gaussian noise with zero mean and variance $\sigma_c^2$. 
We assume i.i.d Rayleigh fading channels, i.e., $h_{t}^{i} \sim \mc{CN}(0,\sigma_h^2), \forall i \in [m]$. 
We also consider the following transmit power constraint for each client at $t$-th communication round:
$\| \z^i_t \|^2 \leq P_t^i, \forall i \in [m]$, $\forall t$,
where $P_t^i$ represents the maximum power that client $i$ can transmit.
We assume that the clients do not have perfect CSIT. 
Instead, each client $i$ can estimate its channel and obtain an imperfect CSI in each global round. 
To facilitate our later analysis, we decompose the estimated CSI $\hh_t^i$ of client $i$ in iteration $t$ into the following two parts:
$\hh_t^i = h_t^i + \Delta_t^i$, $\forall i \in [m]$, $\forall t$,
where $\Delta_t^i$ represents the channel estimation error of client $i$ in the $t$-th round, which is assumed to be a random variable with zero mean and variance $\sigt_h^2$.

We note that without perfect CSIT, we cannot completely offset the influence of fading channels. 
However, we will show that, with imperfect CSIT, we can still achieve an acceptable performance by adaptive power control and dynamic local training design. 

\section{The \alg Algorithm} \label{sec: alg}

To mitigate the impact of channel fading, a well-known approach is to invert the channel at the transmitter by leveraging CSIT.
As mentioned in section-\ref{sec: prelim}, each client trains local model via SGD (cf. Eq.~\eqref{sgd}). Once the local training is done, client $i$ computes its signal $\z_t^i$ and transmits it to the server. Similar to our previous algorithm ACPC-OTA-FL\cite{yang22over}, we propose a dynamic power control (PC) scheme that endows an adaptive scaling factor for each client with a common scaling factor for the server. 
Specifically, denote $\beta_t^i \in \mb{C}$ as PC parameter of client $i$ in the $t$-th round and $\beta_t$ as server-side PC factor. 
The transmission signal $\z_t^i$ is designed as:
\begin{equation}
    \z_t^i = \beta_t^i (\x^i_{t, \tau_t^i} - \x^i_{t, 0}). \label{equ: sigz}
\end{equation}
The server receives the aggregated signal over-the-air and then scales it by $\beta_t$. 
Accordingly, the global update can be expressed as:
\begin{align}
    \x_{t+1} &= \x_{t} + \frac{1}{\beta_t}\sum_{i=1}^{m} \z_t^i  + \tilde{\w}_t, \label{decoding}
\end{align}
where $\tilde{\w}_t$ is the equivalent Gaussian noise after scaling, $\tilde{\w}_t \sim \mc{N}(\0, \frac{\sigma_c^2}{\beta_t^2} \mf{I}_d)$. 
The key difference and novelty in this work compared to \cite{yang22over} is the design of the local PC parameter $\beta_t^i$, which will be illustrated in the next sub-section. 
The design maintains the advantages of previous work\cite{yang22over}, where dynamic local steps allow clients to fully exploit its computation resources while satisfying the communication constraints at the same time.
Furthermore, the new design alleviates the impact of fading on both perfect and imperfect CSI cases.

\setlength {\parskip} {0pt} 
\subsection{Perfect CSI at the Clients} \label{subsec:perfect CSI}
We first assume that all clients have perfect CSI. 
In this case, fading can be fully neutralized via the following PC design:
$\beta_t^i = \frac{\beta_t \alpha_i}{\tau^i_t h_t^i}$.
After scaling at the server, the global update remains the same as those of OTA-FL systems without fading. 
Thus, the convergence analysis is the same as in \cite{yang22over}. 

\setlength {\parskip} {0pt} 
\subsection{Imperfect CSI at the Clients} \label{subsec:imperfect CSI}
We now extend the power control to when clients have imperfect CSI. Specifically,
we use the estimated fading coefficient in the local PC factor to offset the channel fading:
\begin{align} \label{precoding}
    \beta_t^i = \frac{\beta_t \alpha_i}{\tau^i_t \hh_t^i}.
\end{align}
Note that we no longer have a perfect alignment at the server due to imperfect CSI. 
Instead, the information of client $i$ is scaled by $\frac{h_t^i}{\hh_t^i}$ in the aggregated signal. 
The CSI estimation error results in an error in each global model update in every iteration.
As a result, the accumulated mismatch will degrade the overall training performance. 
Moreover, $\tau_t^i$ is selected by plugging~(\ref{equ: sigz}) into the transmit power constraint. When deep fading occurs, a large $\tau_t^i$ is required.
Next, we will analyze the convergence performance and  this cumulative error impact theoretically.
We first state three assumptions:
\begin{assum}($L$-Lipschitz Continuous Gradient) \label{a_smooth}
	There exists a constant $L > 0$, such that $ \| \nabla F_i(\x) - \nabla F_i(\y) \| \leq L \| \x - \y \|$, $\forall \x, \y \in \mathbb{R}^d$, and $i \in [m]$.
\end{assum}

\begin{assum}(Unbiased Local Stochastic Gradients and Bounded Variance) \label{a_unbias}
	Let $\xi_i$ be a random local data sample at client $i$.
	The local stochastic gradient is unbiased and has a bounded variance, i.e.,
	$\mathbb{E} [\nabla F_i(\x, \xi_i)] = \nabla F_i(\x)$, $\forall i \in [m]$, and $\mathbb{E} [\| \nabla F_i(\x, \xi_i) -  \nabla F_i(\x) \|^2] \leq \sigma^2$, where the expectation is taken over the local data distribution $\mc{X}_i$.
\end{assum}

\begin{assum}(Bounded Stochastic Gradient) \label{a_bounded}
	There exists a constant $G \geq 0$, such that the norm of each local stochastic gradient is bounded:
	$\mathbb{E} [\| \nabla F_i(\x, \xi_i) \|^2] \leq G^2$, $\forall i \in [m]$.
\end{assum}

The convergence result of \alg is stated below under the assumptions above:

\begin{restatable}[Convergence Rate of \algns] {theorem} {convergence} \label{thm:convergence}
    Let $\{ \x_t \}$ be the global model parameter.
    Under Assumptions~\ref{a_smooth}-~\ref{a_bounded} and a constant learning rate $\eta_t = \eta  \leq \frac{1}{L}, \forall t \in [T]$, it holds that:
    \begin{multline}
        \min_{t \in [T]} \mb{E} \| \nabla F(\x_t) \|^2 \leq \underbrace{\frac{2 \left(F(\x_0) - F(\x_{*}) \right)}{T \eta}}_{\mathrm{optimization \, error}} + \underbrace{\frac{L \sigma_c^2}{ \eta \beta^2}}_{\substack{\mathrm{channel \,
        noise} \\ \mathrm{error}}}   \nonumber \\
        + \underbrace{2 m L^2 \eta^2 G^2 \sum_{i=1}^m (\alpha_i)^2 \left(\tau_i\right)^2}_{\mathrm{local \, update \, error}} + \underbrace{  L \eta \sigma^2 \frac{1}{T} \sum_{t=0}^{T-1} \sum_{i=1}^{m} \alpha_i^2 \mb{E}_t  \bigg\| \frac{h_t^i}{\hh_t^i} \bigg\|^2}_{\mathrm{statistical \, error}} \nonumber \\
        + \underbrace{2 m G^2  \frac{1}{T} \sum_{t=0}^{T-1} \sum_{i=1}^m (\alpha_i)^2 \mb{E}_t  \bigg\| 1 - \frac{h_t^i}{\hh_t^i} \bigg\|^2}_{\mathrm{channel \, estimation \, error}} \nonumber ,
    \end{multline}
    where $\left(\tau_i\right)^2 = \frac{\sum_{t=0}^{T-1} \left( \tau^i_t \right)^2}{T}$ and $\frac{1}{\bar{\beta}^2} = \frac{1}{T} \sum_{t=0}^{T-1} \frac{1}{\beta_t^2}$.
\end{restatable}

\begin{proof}[Proof Sketch]
The proof of Theorem~\ref{thm:convergence} follows similar steps as in our previous work \cite{yang22over}. We start with one-step function descent and decouple the channel noise term. 
However, the new technical challenge arises from the distortion caused by the CSI estimation error of each client. 
In each round, the global model update aggregates distorted local updates scaled by $\frac{h_t^i}{\hh_t^i}, \forall i \in [m], \forall t \in [T]$. 
Since channel estimation is independent of training, we can decouple the expectation of this factor as an additional error term when calculating the difference of local SGD updates $(\nabla F_i(\x^i_{t,}) - \frac{h_t^i}{\hh_t^i} \nabla F_i(\x^i_{t, k}))$, which  will become an extra scaled factor in statistical noise. 
Due to space limitation, we relegate the full proof to Appendix~\ref{sec:proof}.
\end{proof}

Theorem~\ref{thm:convergence} indicates five sources of errors that affect the convergence rate: 
1) the optimization error depending on the initial guess $\x_0$; 
2) the statistical error due to stochastic gradient noise; 
3) channel noise error from the noisy OTA transmissions;
4) local update error from local update steps coupled with data heterogeneity; 
and 5) channel estimation error due to imperfect CSI. 
Compared to the convergence analysis with perfect CSI in\cite{yang22over}, the additional errors stem from the imperfect CSI at each client. 
The CSI estimation errors in each global iteration accumulate and also contribute to statistical error, which are coupled with bounded local gradient variance and data heterogeneity. 
With perfect CSI, we can fully neutralize this effect. However, with imperfect CSI, we can only partially counter the fading effect.

Note that when we have perfect CSI, i.e., $h_t^i = \hh_t^i$, the accumulated channel estimation error will disappear and the statistical error will not be influenced by $\frac{h_t^i}{\hh_t^i}$. 
This matches our previous convergence analysis of Gaussian MAC. However, when we consider imperfect CSI with Gaussian estimation noise, the convergence upper bound will diverge. 
Yet, in practice, the channel estimation error is a small perturbation. 
Similar to \cite{zhu2020one}, we can use the Taylor expansion to yield the following relation: $\frac{h_t^i}{\hh_t^i} = \frac{1}{1 + \frac{\Delta_t^i}{h_t^i}}  = 1 - \frac{\Delta_t^i}{h_t^i} + \mc{O}( (\frac{\Delta_t^i}{h_t^i})^2)$. By ignoring the higher order terms,  we have the following result:

\begin{restatable}{corollary} {convergence_rate} \label{cor:convergence}
Let $|\Delta_t^i| \ll |h_t^i|, \forall t\in[T], \forall i \in [m]$, $h_{m} = \mathop{min}\nolimits_{t \in [T], i \in[m]}\{|h_t^i|\}$, the convergence rate of {\alg} is bounded. The statistical error and channel estimation error are bounded by:
\begin{multline}
        L \eta \sigma^2 \frac{1}{T} \sum_{t=0}^{T-1} \sum_{i=1}^{m} \alpha_i^2 \mb{E}_t  \bigg\| \frac{h_t^i}{\hh_t^i} \bigg\|^2  \leq L \eta \sigma^2 \sum_{i=1}^{m} \alpha_i^2 \left( 1 + \frac{\sigt_h^2}{h_{m}^2}\right), \nonumber \\
        2 m G^2  \frac{1}{T} \sum_{t=0}^{T-1} \sum_{i=1}^m (\alpha_i)^2 \mb{E}_t  \bigg\| 1 - \frac{h_t^i}{\hh_t^i} \bigg\|^2 \nonumber \leq 2 m G^2 \sum_{i=1}^m (\alpha_i)^2 \frac{\sigt_h^2}{h_{m}^2}.
\end{multline}

\end{restatable}

Finally, we note that we can further extend our results to fast fading channels. When channel states change quickly over time, it is hard to estimate the instantaneous channel gain in each coherent channel duration. 
However, we can obtain the distribution of fading coefficient. 
By replacing the estimated CSI with the expectation of CSI, the aggregated signal remains the same and the convergence results in Theorem~\ref{thm:convergence} still hold.


\section{Numerical Results} \label{sec: exp}

To verify the effectiveness and robustness of \alg, we conduct numerical experiments by using logistic regression for classification tasks on the MNIST dataset~\cite{lecun1998gradient}. 
The experimental setup follows from~\cite{ yang22over}, where data is equally distributed to $m=10$ clients based on labels. 
We use parameter $p$ to represent data heterogeneity level, where $p=10$ means i.i.d. and the rest values are non-i.i.d cases. 
We consider standard i.i.d. Rayleigh fading channels, i.e. $h_t^i \sim \mc{CN}(0,1), \forall t \in [T], \forall i \in [M]$. 
We simulate channel estimation error as a complex Gaussian variable, i.e., $\Delta_t^i \sim \mc{CN}(0,0.1)$, $\forall t \in [T]$, $\forall i \in [M]$. 
The maximum SNR 
is set to \unit[$-1$]{dB}, \unit[$10$]{dB}, \unit[$20$]{dB}.

\begin{table}[tb]
\caption{Logistic regression test accuracy (\%) for \alg compared with COTAF and FedAvg on the MNIST dataset with different non-i.i.d. index $p$ for SNR=$10$. `` / '' means that the algorithm does not converge.}
\label{table:test_accuracy SNR=10}
\centering
\begin{tabular}{|c|c|ccc|}
\hline
\multicolumn{1}{|c|}{\multirow{2}{*}{\textbf{Non-IID Level}}} & \multicolumn{1}{c|}{\multirow{2}{*}{\textbf{Algorithm}}} & \multicolumn{3}{c|}{\textbf{Communication Model}}                                        \\ \cline{3-5} 
\multicolumn{1}{|c|}{}                               & \multicolumn{1}{c|}{}                           & \multicolumn{1}{c|}{Imperfect}  & \multicolumn{1}{c|}{Perfect}  & No Fading  \\ \hline
\multirow{3}{*}{$p=1$}                                 & {\cellcolor[gray]{.9}{\bf \alg}}                                     & \multicolumn{1}{c|}{\cellcolor[gray]{.9} {\bf 85.87}} & \multicolumn{1}{c|}{\cellcolor[gray]{.9} {\bf 87.46}} & \cellcolor[gray]{.9} {\bf 89.08} \\ \cline{2-5} 
                                                     & COTAF                                           & \multicolumn{1}{c|}{/} & \multicolumn{1}{c|}{63.96} & 65.54 \\ \cline{2-5} 
                                                     & FedAvg                                          & \multicolumn{1}{c|}{/} & \multicolumn{1}{c|}{69.64} & 68.08 \\ \hline
\multirow{3}{*}{$p=2$}                                 & {\cellcolor[gray]{.9}{\bf \alg}}                                     & \multicolumn{1}{c|}{\cellcolor[gray]{.9} {\bf 86.45}} & \multicolumn{1}{c|}{\cellcolor[gray]{.9} {\bf 89.07}} & \cellcolor[gray]{.9} {\bf 89.58} \\ \cline{2-5} 
                                                     & COTAF                                           & \multicolumn{1}{c|}{/} & \multicolumn{1}{c|}{77.47} & 78.80 \\ \cline{2-5} 
                                                     & FedAvg                                          & \multicolumn{1}{c|}{51.96} & \multicolumn{1}{c|}{79.42} & 78.03 \\ \hline
\multirow{3}{*}{$p=5$}                                 & {\cellcolor[gray]{.9}{\bf \alg}}                                     & \multicolumn{1}{c|}{\cellcolor[gray]{.9} {\bf 89.27}} & \multicolumn{1}{c|}{\cellcolor[gray]{.9} {\bf 91.07}} & \cellcolor[gray]{.9}{\bf 90.64} \\ \cline{2-5} 
                                                     & COTAF                                           & \multicolumn{1}{c|}{/} & \multicolumn{1}{c|}{85.96} & 86.52 \\ \cline{2-5} 
                                                     & FedAvg                                          & \multicolumn{1}{c|}{59.49} & \multicolumn{1}{c|}{82.19} & 82.84 \\ \hline
\multirow{3}{*}{$p=10$}                                & {\cellcolor[gray]{.9}{\bf \alg}}                                     & \multicolumn{1}{c|}{\cellcolor[gray]{.9}{\bf90.06}} & \multicolumn{1}{c|}{\cellcolor[gray]{.9}{\bf 91.19}} & \cellcolor[gray]{.9}{\bf{90.75}} \\ \cline{2-5} 
                                                     & COTAF                                           & \multicolumn{1}{c|}{/} & \multicolumn{1}{c|}{91.04} & 91.08 \\ \cline{2-5} 
                                                     & FedAvg                                          & \multicolumn{1}{c|}{61.79} & \multicolumn{1}{c|}{84.85} & 84.94 \\ \hline
\end{tabular}
\vspace{-.1in}
\end{table}

\begin{table}[tb]
\caption{Logistic regression test Accuracy (\%) for \alg compared with COTAF and FedAvg on the MNIST dataset with different signal-to-noise ratios when Non-IID index $p=2$. `` / '' means that the algorithm does not converge.}
\label{table:test_accuracy p=2}
\centering
\begin{tabular}{|c|c|ccc|}
\hline
\multicolumn{1}{|c|}{\multirow{2}{*}{\textbf{SNR}}} & \multicolumn{1}{c|}{\multirow{2}{*}{\textbf{Algorithm}}} & \multicolumn{3}{c|}{\textbf{Communication Model}}                                        \\ \cline{3-5} 
\multicolumn{1}{|c|}{}                               & \multicolumn{1}{c|}{}                           & \multicolumn{1}{c|}{Imperfect}  & \multicolumn{1}{c|}{Perfect}  & No Fading  \\ \hline
\multirow{3}{*}{$\mathsf{SNR}=-1$}                                 & {\cellcolor[gray]{.9}{\bf \alg}}                                     & \multicolumn{1}{c|}{\cellcolor[gray]{.9} {\bf 79.54}} & \multicolumn{1}{c|}{\cellcolor[gray]{.9} {\bf 82.88}} & \cellcolor[gray]{.9} {\bf 81.89} \\ \cline{2-5} 
                                                     & COTAF                                           & \multicolumn{1}{c|}{/} & \multicolumn{1}{c|}{61.33} & 63.59 \\ \cline{2-5} 
                                                     & FedAvg                                          & \multicolumn{1}{c|}{/} & \multicolumn{1}{c|}{73.17} & 71.55 \\ \hline
\multirow{3}{*}{$\mathsf{SNR}=10$}                                 & {\cellcolor[gray]{.9}{\bf \alg}}                                     & \multicolumn{1}{c|}{\cellcolor[gray]{.9} {\bf 86.45}} & \multicolumn{1}{c|}{\cellcolor[gray]{.9} {\bf 89.07}} & \cellcolor[gray]{.9} {\bf 89.58} \\ \cline{2-5} 
                                                     & COTAF                                           & \multicolumn{1}{c|}{/} & \multicolumn{1}{c|}{77.47} & 78.80 \\ \cline{2-5} 
                                                     & FedAvg                                          & \multicolumn{1}{c|}{51.96} & \multicolumn{1}{c|}{79.42} & 78.03 \\ \hline
\multirow{3}{*}{$\mathsf{SNR}=20$}                                 & {\cellcolor[gray]{.9}{\bf \alg}}                                     & \multicolumn{1}{c|}{\cellcolor[gray]{.9} {\bf 87.10}} & \multicolumn{1}{c|}{\cellcolor[gray]{.9} {\bf 90.17}} & \cellcolor[gray]{.9}{\bf 90.43} \\ \cline{2-5} 
                                                     & COTAF                                           & \multicolumn{1}{c|}{/} & \multicolumn{1}{c|}{86.10} & 86.57 \\ \cline{2-5} 
                                                     & FedAvg                                          & \multicolumn{1}{c|}{63.36} & \multicolumn{1}{c|}{79.32} & 79.86 \\ \hline
\end{tabular}
\vspace{-.1in}
\end{table}

 In this experiment, we only focus on imperfect CSIT. 
 From table~\ref{tab:related_work}, OBDA is a digital method, CA-DSGD considers compression, EADDS has a computation energy constraint.
 All of these methods have different perspectives that we do not consider. 
 For a fair comparison, we compare {\alg} with COTAF~\cite{sery2021over} and FedAvg~\cite{mcmahan2017communication} in different communication scenarios. 
 In an imperfect CSI case, we let each client use an estimated channel to inverse the fading effect in both COTAF and FedAvg. 
 Specifically, the transmitted signal from client $i$ is scaled by $\frac{1}{\hh_t^i}$.

From the test accuracy results in Table~\ref{table:test_accuracy SNR=10} and~\ref{table:test_accuracy p=2}, we observe that the performance of perfect CSI is the same as no fading case, which implies that the increased local steps due to the inverse fading channel gain do not introduce bias to the global model. 
In the imperfect CSI scenario, our \alg algorithm can still achieve an acceptable test accuracy, albeit worse than the perfect CSI due to the CSI estimation error.
Note that \alg outperforms COTAF and FedAvg significantly. 
COTAF and FedAvg fail to converge, suggesting that the impacts of CSI estimation error could be significant.
Our \alg algorithm allows clients to choose local steps dynamically, where the joint computation and communication design can mitigate the effect of CSI estimation error and maintain system robustness.


\section{Conclusion} \label{sec: conclusion}

In this paper, we have proposed a new adaptive OTA-FL algorithm called \alg. The proposed algorithm adapts to wireless channel fading with channel inversion at each client utilizing estimated CSI. 
We have considered the practical scenario where the clients only have imperfect CSI. 
We have studied the convergence performance of \alg with imperfect CSI and quantified the impact of CSI estimation error. 
We have demonstrated the effectiveness and robustness of the joint communication and computation design under data and system heterogeneity.




\bibliographystyle{IEEEtran}{}
\bibliography{BIB/FederatedLearning, BIB/SGD, BIB/Yang, BIB/OTAFL, BIB/Experiments, BIB/ImperfectCSI}

\onecolumn
\allowdisplaybreaks

\section{Proof}
\label{sec:proof}
\textbf{Theorem 1} (Convergence Rate of \algns)
Let $\{ \x_t \}$ be the global model parameter.
    Under Assumptions~\ref{a_smooth}-~\ref{a_bounded} and a constant learning rate $\eta_t = \eta \leq \frac{1}{L}, \forall t \in [T]$, it holds that:
    \begin{multline}
        \min_{t \in [T]} \mb{E} \| \nabla F(\x_t) \|^2 \leq \underbrace{\frac{2 \left(F(\x_0) - F(\x_{*}) \right)}{T \eta}}_{\mathrm{optimization \, error}} + \underbrace{\frac{L \sigma_c^2}{ \eta \beta^2}}_{\substack{\mathrm{channel \,
        noise} \\ \mathrm{error}}}  
        + \underbrace{2 m L^2 \eta^2 G^2 \sum_{i=1}^m (\alpha_i)^2 \left(\tau_i\right)^2}_{\mathrm{local \, update \, error}} + \underbrace{  L \eta \sigma^2 \frac{1}{T} \sum_{t=0}^{T-1} \sum_{i=1}^{m} \alpha_i^2 \mb{E}_t  \bigg\| \frac{h_t^i}{\hh_t^i} \bigg\|^2}_{\mathrm{statistical \, error}} \nonumber \\
        + \underbrace{2 m G^2  \frac{1}{T} \sum_{t=0}^{T-1} \sum_{i=1}^m (\alpha_i)^2 \mb{E}_t  \bigg\| 1 - \frac{h_t^i}{\hh_t^i} \bigg\|^2}_{\mathrm{channel \, estimation \, error}}  ,
    \end{multline}
    where $\left(\tau_i\right)^2 = \frac{\sum_{t=0}^{T-1} \left( \tau^i_t \right)^2}{T}$ and $\frac{1}{\bar{\beta}^2} = \frac{1}{T} \sum_{t=0}^{T-1} \frac{1}{\beta_t^2}$.

\begin{proof}

\begin{align}
    \x_{t+1} - \x_{t} &= \sum_{i=1}^{m} \frac{\beta_t^i}{\beta_t} h_t^i \left(\x^i_{t, \tau^i_t} - \x^i_{t, 0}\right) + \tilde{\w}_t \\
    &= \sum_{i=1}^{m} \frac{\alpha_i}{\tau^i_t} \frac{h_t^i}{\hh_t^i} \left(\x^i_{t, \tau^i_t} - \x^i_{t, 0}\right) + \tilde{\w}_t \\
    &= - \sum_{i=1}^{m} \frac{\alpha_i}{\tau^i_t}\frac{h_t^i}{\hh_t^i} \eta_t \sum_{k=0}^{\tau^i_t-1} \left(\nabla F_i(\x^i_{t, k}, \xi^i_{t, k})\right) + \tilde{\w}_t
\end{align}

We first take expectation conditioned on $\x_t$.
There exists three sources of randomness: stochastic gradient noise, channel noise, and imperfect CSI estimation noise; but we assume they are independent.
According to assumption 1,
\begin{align}
    &\mb{E}_t [F(\x_{t+1})] - F(\x_t) \leq \left< \nabla F(\x_t), \mb{E}_t \left[\x_{t+1} - \x_t \right] \right> + \frac{L}{2} \mb{E}_t \left[\| \x_{t+1} - \x_t \|^2 \right] \\
    &= - \left< \nabla F(\x_t), \mb{E}_t \left[\sum_{i=1}^{m} \frac{\alpha_i}{\tau^i_t} \frac{h_t^i}{\hh_t^i} \eta_t \sum_{k=0}^{\tau^i_t-1} \left(\nabla F_i(\x^i_{t, k})\right)\right] \right> + \frac{L}{2} \mb{E}_t \left[\bigg\| - \sum_{i=1}^{m} \eta_t\frac{\alpha_i }{\tau^i_t}\frac{h_t^i}{\hh_t^i} \sum_{k=0}^{\tau^i_t-1} \left(\nabla F_i(\x^i_{t, k}, \xi^i_{t, k})\right)  + \tilde{\w}_t \bigg\|^2 \right] \\
    &= - \eta_t \| \nabla F(\x_t) \|^2 + \eta_t \left< \nabla F(\x_t), \nabla F(\x_t) - \mb{E}_t \left[\sum_{i=1}^{m} \frac{\alpha_i}{\tau^i_t}\frac{h_t^i}{\hh_t^i} \sum_{k=0}^{\tau^i_t-1} \left(\nabla F_i(\x^i_{t, k})\right) \right] \right> \\
    &+ \frac{L}{2} \mb{E}_t \left[\bigg\| - \sum_{i=1}^{m} \frac{\alpha_i \eta_t}{\tau^i_t} \frac{h_t^i}{\hh_t^i} \sum_{k=0}^{\tau^i_t-1} \left(\nabla F_i(\x^i_{t, k}, \xi^i_{t, k})\right) + \tilde{\w}_t \bigg\|^2 \right] \\
    &=- \frac{1}{2} \eta_t \| \nabla F(\x_t) \|^2 - \frac{1}{2} \eta_t \mb{E}_t \left[\bigg\| \sum_{i=1}^{m} \frac{\alpha_i}{\tau^i_t} \frac{h_t^i}{\hh_t^i} \sum_{k=0}^{\tau^i_t-1} \left(\nabla F_i(\x^i_{t, k})\right) \bigg\|^2 \right] + \frac{1}{2} \eta_t \mb{E}_t \bigg\| \nabla F(\x_t) - \sum_{i=1}^{m} \frac{\alpha_i}{\tau^i_t} \frac{h_t^i}{\hh_t^i}\sum_{k=0}^{\tau^i_t-1} \left(\nabla F_i(\x^i_{t, k})\right) \bigg\|^2 \\
    &+ \frac{L}{2} \mb{E}_t \left[\bigg\| - \sum_{i=1}^{m} \frac{\alpha_i \eta_t}{\tau^i_t}\frac{h_t^i}{\hh_t^i} \sum_{k=0}^{\tau^i_t-1} \left(\nabla F_i(\x^i_{t, k}, \xi^i_{t, k})\right) + \tilde{\w}_t \bigg\|^2 \right] \\
    &=- \frac{1}{2} \eta_t \| \nabla F(\x_t) \|^2 - \frac{1}{2} \eta_t \mb{E}_t \left[\bigg\| \sum_{i=1}^{m} \frac{\alpha_i}{\tau^i_t} \frac{h_t^i}{\hh_t^i} \sum_{k=0}^{\tau^i_t-1} \left(\nabla F_i(\x^i_{t, k})\right) \bigg\|^2 \right] + \frac{1}{2} \eta_t \mb{E}_t \bigg\| \sum_{i=1}^{m} \frac{\alpha_i}{\tau^i_t} \sum_{k=0}^{\tau^i_t-1} \left(\nabla F_i(\x_t) - \frac{h_t^i}{\hh_t^i} \nabla F_i(\x^i_{t, k})\right) \bigg\|^2 \\
    &+ \frac{L \eta_t^2}{2} \mb{E}_t \left[\bigg\| \sum_{i=1}^{m} \frac{\alpha_i}{\tau^i_t} \frac{h_t^i}{\hh_t^i}\sum_{k=0}^{\tau^i_t-1} \left(\nabla F_i(\x^i_{t, k}, \xi^i_{t, k})\right) \bigg\|^2 \right] + \frac{L \sigma_c^2}{2\beta_t^2} \\
    &\leq - \frac{1}{2} \eta_t \| \nabla F(\x_t) \|^2 + \frac{1}{2} \eta_t \mb{E}_t \bigg\| \sum_{i=1}^{m} \frac{\alpha_i}{\tau^i_t} \sum_{k=0}^{\tau^i_t-1} \left(\nabla F_i(\x_t) - \frac{h_t^i}{\hh_t^i}\nabla F_i(\x^i_{t, k})\right) \bigg\|^2 \\
    &+ \frac{L \eta_t^2}{2} \mb{E}_t \left[\bigg\| \sum_{i=1}^{m} \frac{\alpha_i}{\tau^i_t}\frac{h_t^i}{\hh_t^i} \sum_{k=0}^{\tau^i_t-1} \left(\nabla F_i(\x^i_{t, k}, \xi^i_{t, k})\right) - \sum_{i=1}^{m} \frac{\alpha_i}{\tau^i_t} \frac{h_t^i}{\hh_t^i}\sum_{k=0}^{\tau^i_t-1} \left(\nabla F_i(\x^i_{t, k})\right) \bigg\|^2 \right] + \frac{L \sigma_c^2}{2\beta_t^2} \\
    &\leq - \frac{1}{2} \eta_t \| \nabla F(\x_t) \|^2 + \frac{1}{2} \eta_t \mb{E}_t \bigg\| \sum_{i=1}^{m} \frac{\alpha_i}{\tau^i_t} \sum_{k=0}^{\tau^i_t-1} \left(\nabla F_i(\x_t) - \frac{h_t^i}{\hh_t^i} \nabla F_i(\x^i_{t, k})\right) \bigg\|^2 \\
    &+ \frac{L \eta_t^2}{2} \sum_{i=1}^{m} \mb{E}_t \left[\bigg\| \frac{\alpha_i}{\tau^i_t}\frac{h_t^i}{\hh_t^i}  \sum_{k=0}^{\tau^i_t-1} \left(\nabla F_i(\x^i_{t, k}, \xi^i_{t, k}) - \nabla F_i(\x^i_{t, k})\right) \bigg\|^2 \right] + \frac{L \sigma_c^2}{2\beta_t^2} \label{ineq: Lsmooth} 
\end{align}

The first inequality holds if $\eta_t \leq \frac{1}{L}$. Because channel estimation is independent of learning,
\begin{align}
    &\frac{ L \eta_t^2}{2} \sum_{i=1}^{m} \mb{E}_t \left[\bigg\| \frac{\alpha_i}{\tau^i_t}\frac{h_t^i}{\hh_t^i}  \sum_{k=0}^{\tau^i_t-1} \left(\nabla F_i(\x^i_{t, k}, \xi^i_{t, k}) - \nabla F_i(\x^i_{t, k})\right) \bigg\|^2 \right] \\
    &\leq
    \frac{ L \eta_t^2}{2}  \sum_{i=1}^m \frac{(\alpha_i)^2}{(\tau^i_t)^2} \mb{E}_t  \bigg\| \frac{h_t^i}{\hh_t^i} \bigg\|^2\mb{E}_t \bigg\|  \sum_{k=0}^{\tau^i_t-1} \left(\nabla F_i(\x^i_{t, k}, \xi^i_{t, k}) - \nabla F_i(\x^i_{t, k})\right) \bigg\|^2 \\
    &\leq \frac{ L \eta_t^2}{2} \sum_{i=1}^m \frac{(\alpha_i)^2}{\tau^i_t} \mb{E}_t  \bigg\| \frac{h_t^i}{\hh_t^i} \bigg\|^2 \sum_{k=0}^{\tau^i_t-1} \mb{E}_t \bigg\| \left(\nabla F_i(\x^i_{t, k}, \xi^i_{t, k}) - \nabla F_i(\x^i_{t, k})\right) \bigg\|^2 \\
    &\leq \frac{ L \eta_t^2}{2} \sum_{i=1}^m (\alpha_i)^2 \mb{E}_t  \bigg\| \frac{h_t^i}{\hh_t^i} \bigg\|^2 \sigma^2
     \label{ineq: boundedvar}
\end{align}

Apply Jensen's Inequality to assumption 2 and 3, we have $\bigg\| \nabla F_i(\x^i_{t, k}) \bigg\|^2 \leq G^2$. 

\begin{align}
    &\frac{1}{2} \eta_t \mb{E}_t \bigg\| \sum_{i=1}^{m} \frac{\alpha_i}{\tau^i_t} \sum_{k=0}^{\tau^i_t-1} \left(\nabla F_i(\x_t) - \frac{h_t^i}{\hh_t^i}\nabla F_i(\x^i_{t, k})\right) \bigg\|^2 \leq
    \frac{1}{2} \eta_t m \sum_{i=1}^m \frac{(\alpha_i)^2}{(\tau^i_t)^2} \mb{E}_t \bigg\|  \sum_{k=0}^{\tau^i_t-1} \left(\nabla F_i(\x_t) - \frac{h_t^i}{\hh_t^i} \nabla F_i(\x^i_{t, k})\right) \bigg\|^2 \\
    &\leq \frac{1}{2} \eta_t m \sum_{i=1}^m \frac{(\alpha_i)^2}{\tau^i_t} \sum_{k=0}^{\tau^i_t-1} \mb{E}_t \bigg\| \left(\nabla F_i(\x_t) - \frac{h_t^i}{\hh_t^i}\nabla F_i(\x^i_{t, k})\right) \bigg\|^2 \\
    &= \frac{1}{2} \eta_t m \sum_{i=1}^m \frac{(\alpha_i)^2}{\tau^i_t} \sum_{k=0}^{\tau^i_t-1} \mb{E}_t \bigg\| \left(\nabla F_i(\x_t) -\nabla F_i(\x^i_{t, k}) + \nabla F_i(\x^i_{t, k}) - \frac{h_t^i}{\hh_t^i}\nabla F_i(\x^i_{t, k})\right) \bigg\|^2 \\
    &\leq \frac{1}{2} \eta_t m \sum_{i=1}^m \frac{(\alpha_i)^2}{\tau^i_t} \sum_{k=0}^{\tau^i_t-1} \mb{E}_t  \left(2\bigg\| \nabla F_i(\x_t) -\nabla F_i(\x^i_{t, k}) \bigg\|^2 + 2 \bigg\| \nabla F_i(\x^i_{t, k}) - \frac{h_t^i}{\hh_t^i}\nabla F_i(\x^i_{t, k}) \bigg\|^2 \right) \\
    &\leq \eta_t m \sum_{i=1}^m \frac{(\alpha_i)^2}{\tau^i_t} \sum_{k=0}^{\tau^i_t-1} \left( L^2 \mb{E}_t \bigg\|\x_t - \x^i_{t, k} \bigg\|^2 + \mb{E}_t \bigg\|1 - \frac{h_t^i}{\hh_t^i}\bigg\|^2 \mb{E}_t \bigg\| \nabla F_i(\x^i_{t, k}) \bigg\|^2 \right)\\
    &\leq \eta_t m \sum_{i=1}^m \frac{(\alpha_i)^2}{\tau^i_t} \sum_{k=0}^{\tau^i_t-1}\left( \eta_t^2 L^2 \mb{E}_t \bigg\|\sum_{j=0}^{k} \nabla F_i(\x^i_{t, j}, \xi^i_{t, j})  \bigg\|^2 + \mb{E}_t \bigg\| 1 - \frac{h_t^i}{\hh_t^i}\bigg\|^2 \mb{E}_t \bigg\| \nabla F_i(\x^i_{t, k}) \bigg\|^2 \right) \\
    &\leq \eta_t m \sum_{i=1}^m \frac{(\alpha_i)^2}{\tau^i_t} \sum_{k=0}^{\tau^i_t-1} \left( \eta_t^2 L^2 k^2 G^2  + \mb{E}_t \bigg\|1- \frac{h_t^i}{\hh_t^i}\bigg\|^2 G^2 \right) \\
    &\leq \eta_t^3 m L^2 \sum_{i=1}^m (\alpha_i)^2 (\tau^i_t)^2 G^2 + \eta_t m \sum_{i=1}^m (\alpha_i)^2 \mb{E}_t \bigg\|1- \frac{h_t^i}{\hh_t^i}\bigg\|^2 G^2
    \label{ineq: variance}
\end{align}

Plugging inequality~\eqref{ineq: boundedvar} and~\eqref{ineq: variance} into \eqref{ineq: Lsmooth}, we have 
\begin{align}
    &\mb{E}_t [F(\x_{t+1})] - F(\x_t) \leq \left< \nabla F(\x_t), \mb{E}_t \left[\x_{t+1} - \x_t \right] \right> + \frac{L}{2} \mb{E}_t \left[\| \x_{t+1} - \x_t \|^2 \right] \\
    &\leq - \frac{1}{2} \eta_t \| \nabla F(\x_t) \|^2 + \eta_t^3 m L^2 \sum_{i=1}^m (\alpha_i)^2 (\tau^i_t)^2 G^2 + \eta_t m \sum_{i=1}^m (\alpha_i)^2 \mb{E}_t \bigg\|1- \frac{h_t^i}{\hh_t^i}\bigg\|^2 G^2 + \frac{L \eta_t^2}{2} \sum_{i=1}^m (\alpha_i)^2 \mb{E}_t  \bigg\| \frac{h_t^i}{\hh_t^i} \bigg\|^2 \sigma^2 + \frac{L \sigma_c^2}{2\beta_t^2}
\end{align}

Rearranging and telescoping:
\begin{align}
    \frac{1}{T} \sum_{t=0}^{T-1} \eta_t \mb{E}_t \| \nabla F(\x_t) \|^2 &\leq \frac{2 \left(F(\x_0) - F(\x_T) \right)}{T} +  2 m L^2 G^2 \frac{1}{T} \sum_{t=0}^{T-1} \eta_t^3 \sum_{i=1}^m (\alpha_i)^2 \left( \tau^i_t \right)^2 + L \sigma_c^2 \frac{1}{T} \sum_{t=0}^{T-1} \frac{1}{\beta_t^2}  \\ 
    & + 2 m G^2 \frac{1}{T} \sum_{t=0}^{T-1} \eta_t \sum_{i=1}^m (\alpha_i)^2 \mb{E}_t  \bigg\| 1 - \frac{h_t^i}{\hh_t^i} \bigg\|^2  +  L \frac{1}{T} \sum_{t=0}^{T-1} \eta_t^2 \sum_{i=1}^m (\alpha_i)^2 \mb{E}_t  \bigg\| \frac{h_t^i}{\hh_t^i} \bigg\|^2 \sigma^2
\end{align}

Let $\eta_t = \eta$ be constant learning rate, $\left(\tau_i\right)^2 = \frac{\sum_{t=0}^{T-1} \left( \tau^i_t \right)^2}{T}$ then we have:
\begin{align}
    \frac{1}{T} \sum_{t=0}^{T-1} \mb{E} \| \nabla F(\x_t) \|^2 &\leq \frac{2 \left(F(\x_0) - F(\x_T) \right)}{T \eta} + 2 m L^2 \eta^2 G^2 \sum_{i=1}^m (\alpha_i)^2 \left(\tau_i\right)^2 + \frac{L \sigma_c^2}{\eta \beta^2} + 2 m G^2  \frac{1}{T} \sum_{t=0}^{T-1} \sum_{i=1}^m (\alpha_i)^2 \mb{E}_t  \bigg\| 1 - \frac{h_t^i}{\hh_t^i} \bigg\|^2 +\\
    & L \eta \frac{1}{T} \sum_{t=0}^{T-1} \sum_{i=1}^m (\alpha_i)^2 \mb{E}_t  \bigg\| \frac{h_t^i}{\hh_t^i} \bigg\|^2 \sigma^2,
\end{align}
where $\frac{1}{\bar{\beta}^2} = \frac{1}{T} \sum_{t=0}^{T-1} \frac{1}{\beta_t^2}$.
\end{proof}

\end{document}